\documentclass[journal,9.7pt]{IEEEtran}

\usepackage{amsfonts}
\usepackage{slashbox}
\usepackage{amsmath}
\usepackage{multirow}
\usepackage{graphicx}
\usepackage{amsfonts}
\usepackage{amsmath,epsfig}
\usepackage{mathbbold}
\usepackage{stfloats}
\usepackage{amssymb}
\usepackage{url}
\usepackage{bm}
\usepackage{cite}
\usepackage{amsmath}
\usepackage{amssymb}
\usepackage{latexsym}
\usepackage{multirow}
\usepackage{epsfig}
\usepackage{graphics}
\usepackage{mathrsfs}
\usepackage{algorithmic}
\usepackage{algorithm}
\usepackage{subfigure}
\usepackage{slashbox}
\usepackage[all]{xy}

\usepackage{pifont}
\usepackage{bbding}
\usepackage{stmaryrd}
\usepackage{amssymb}
\usepackage{amsfonts}
\usepackage{epic}
\usepackage{stfloats}
\usepackage{latexsym}
\usepackage{epstopdf}
\usepackage{epic}
\usepackage{bm}
\usepackage{xcolor}
\usepackage{mathrsfs}
\usepackage{pifont}
\usepackage{bbding}
\usepackage{amsmath,epsfig}
\usepackage{mathbbold}
\usepackage{stmaryrd}
\usepackage{amssymb}
\usepackage{amsfonts}
\usepackage{epic}
\usepackage{graphicx}
\usepackage{subfigure}
\usepackage{enumerate}
\usepackage{stfloats}
\usepackage{latexsym}
\usepackage{epstopdf}
\usepackage{epic}
\usepackage{multirow}
\usepackage{stfloats}
\usepackage{bm}

\usepackage{booktabs}
\usepackage{color}
\usepackage{amsthm}

\begin{document}

\newcommand{\reffig}[1]{Fig. \ref{#1}}
\newcommand{\figref}[1]{\figurename~\ref{#1}}

\newtheorem{definition}{Definition}
\newtheorem{theorem}{Theorem}
\newtheorem{corollary}{Corollary}
\newtheorem{proposition}{Proposition}
\newtheorem{lemma}{Lemma}
\newtheorem{property}{Property}
\newtheorem{remark}{Remark}
\renewcommand{\algorithmicrequire}{\textbf{Input:}}  
\renewcommand{\algorithmicensure}{\textbf{Output:}}  

\title{ \huge{Spectral and Energy Efficient Wireless Powered IoT Networks: NOMA or TDMA?}}

\author{\IEEEauthorblockN{Qingqing Wu,   Wen Chen,  Derrick Wing Kwan Ng, and Robert Schober
\thanks{Qingqing Wu and Wen Chen  are with Department of Electronic Engineering, Shanghai Jiao Tong University, email: \{wu.qq,  wenchen\}@sjtu.edu.cn.  Derrick Wing Kwan Ng is with the School of Electrical Engineering and Telecommunications, The University of New South Wales, Australia, email: w.k.ng@unsw.edu.au.  Robert Schober is with the Institute for Digital Communications, Friedrich-Alexander-University Erlangen-N\"urnberg (FAU), email: robert.schober@fau.de.   }
}  }

\maketitle
\begin{abstract}
Wireless powered communication networks (WPCNs), where multiple energy-limited devices first harvest energy in the downlink and then transmit information in the uplink, have been envisioned as a promising solution for the future Internet-of-Things (IoT). Meanwhile, non-orthogonal multiple access (NOMA) has been proposed
to improve the system spectral efficiency (SE) of the fifth-generation (5G) networks by allowing concurrent transmissions of multiple users in the same spectrum. As such, NOMA has been recently considered for the uplink of WPCNs based IoT networks with a massive number of devices.
However,  simultaneous transmissions  in NOMA  may also incur more transmit energy consumption as well as  circuit energy consumption in practice which is critical for energy constrained IoT devices. As a result,  compared to orthogonal multiple access schemes such as time-division multiple access (TDMA), whether the SE can be improved and/or the total energy consumption can be reduced with NOMA in such a scenario still remains unknown. To answer this question, 
we first derive the optimal time allocations for maximizing the SE of a TDMA-based WPCN (T-WPCN) and a NOMA-based WPCN (N-WPCN), respectively. Subsequently, we analyze the total energy consumption as well as the maximum SE achieved by these two networks. Surprisingly, it is found  that N-WPCN not only consumes more energy, but also is less spectral efficient than T-WPCN. Simulation results verify our theoretical findings and unveil the fundamental performance bottleneck, i.e., ``worst user bottleneck problem'',  in multiuser NOMA systems.

\end{abstract}

\vspace{-0.3cm}
\section{Introduction}
The number of connected devices will skyrocket to $30$ billion by $2025$, giving rise to the well known ``Internet-of-Things (IoT)" \cite{zhang2016fundamental}. With such a huge number of IoT devices, the lifetime of networks becomes a critical issue and the conventional battery based solutions may no longer be sustainable due to the high cost of battery replacement as well as environmental concerns. As a result, wireless power transfer, which enables energy harvesting from ambient radio frequency (RF) signals, is envisioned as a promising solution for powering massive IoT devices \cite{wu2016overview}. However, due to the significant signal attenuation in wireless communication channels, the harvested RF energy at the devices is generally limited. Therefore, how to efficiently utilize the scarce harvested energy becomes particularly crucial for realizing sustainable and scalable IoT networks.  To this end, a ``harvest and then transmit" protocol  is proposed in \cite{ju14_throughput,qing15_wpcn_twc,zhang2017wireless} for wireless powered communication networks (WPCNs), where devices first harvest energy in the downlink (DL) for wireless energy transfer (WET) and then transmit information signals in the uplink (UL) for wireless information transmission (WIT). 

 Meanwhile, non-orthogonal multiple access (NOMA) has been proposed to improve the SE as well as user fairness  by allowing multiple users simultaneously to access the same spectrum. With successive interference cancellation (SIC) performed at the receiver, NOMA has been demonstrated superior to orthogonal multiple access (OMA) in terms of the ergodic sum rate \cite{ding2014performance}. As such, NOMA is recently pursued for UL WIT in WPCNs \cite{diamantoulakis2016wireless,chingoska2016resource}, where the decoding order of the users is exploited to enhance the throughput fairness among users. However,  the conclusions drawn in \cite{ding2014performance} are only applicable for the DL scenario and may not hold for UL IoT networks with energy constrained devices. {Furthermore,  \cite{diamantoulakis2016wireless} and \cite{chingoska2016resource} focus only on improving the system/individual user throughput without considering the total system energy consumption. In fact, a theoretical total energy consumption comparison between NOMA and TDMA is important since the efficiency of WET is generally low in practice.
    Also, the circuit energy consumption of the users is completely ignored in \cite{ju14_throughput,diamantoulakis2016wireless,chingoska2016resource}. However, the circuit power consumption is often comparable to the transmit power and thus important for short-range IoT applications \cite{qing16_wpcn_twc,miao2010energy,
qing1,cui2004energy,qing15tx_rx}, such as wearables devices.} As multiple users access the same spectrum simultaneously in NOMA, the circuit energy consumption of each user increases inevitably, which may contradict a fundamental design requirement of future IoT networks, i.e., ultra low power consumption \cite{wang2017primer}. 
For example, in NOMA-based WPCN (N-WPCN) with a fixed total available harvested energy, if devices consume more energy for operating their circuits than in time-division multiple access (TDMA)-based WPCN (T-WPCN), then less energy will be left for signal transmission. As a result, a natural question arises: Does NOMA improve the SE and/or reduce the total energy consumption of such wireless powered IoT networks in practice compared to TDMA?



\begin{figure}[!t]
\centering
\includegraphics[width=3in]{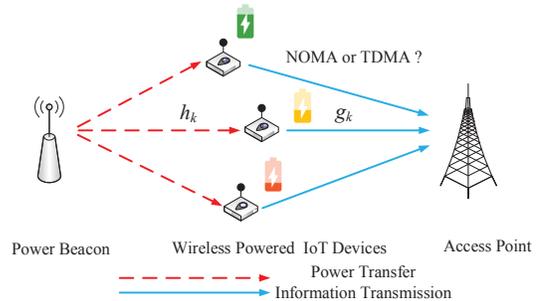}
\caption{System model  of a wireless powered IoT network. }\label{convergence}\vspace{-0.45cm}
\end{figure}

Driven by  the above question, we make the following contributions in this paper. 1) By taking into account the circuit energy consumption, we first derive the optimal time allocation for the SE maximization problem for T-WPCN, based on which,  the corresponding problem for N-WPCN can be cast as the single user case  for T-WPCN; 2) we prove that N-WPCN in general requires a longer DL WET time duration than T-WPCN, which implies that N-WPCN is more energy demanding; 3) we prove that N-WPCN in general achieves a lower SE than  T-WPCN. {Given 2) and 3), NOMA may not be a good candidate for realizing spectral and energy efficient wireless powered IoT networks if the circuit energy consumption is not negligible.}
\section{System Model and Problem Formulation}
\subsection{System Model}

 We consider a WPCN, which consists of one power beacon (PB), $K>1$ wireless-powered IoT devices, and one information access point (AP), as shown in Fig. 1.
 The total available transmission time is denoted by $T_{\max}$.  The ``harvest and then transmit" protocol \cite{ju14_throughput} is adopted where the devices first harvest energy from the signal sent by the PB and then transmit information to the AP. {We note that the ``doubly near-far phenomenon'' \cite{ju14_throughput} can be avoided by using separated PB and AP as in our model \cite{qing15_wpcn_twc,bi2016wireless}.}
To compare the upper bound performance of T-WPCN and N-WPCN,  we assume that perfect channel state information (CSI) is available for resource allocation.
The DL channel gain between the PB and device $k\in\{1,2,\ldots,K\}$, and the UL channel gain between device $k$ and the AP are denoted by $h_k$ and $g_k$, respectively.

During DL WET, the PB broadcasts the energy signal with a constant transmit power $P_{\rm E}$ for time $\tau_0$.
 The energy harvested from the noise and the received UL WIT signals
from other devices are assumed to be negligible,  since both the noise power and device transmit power are much smaller than the transmit power of the PB in practice \cite{ju14_throughput}. Thus, the amount of harvested energy  at device $k$ can be expressed as
\begin{align}\label{eq3}
E^h_k=\eta_kP_{\rm E}h_k\tau_0,
\end{align}
where $\eta_k \in (0,1]$ is the constant energy conversion efficiency of device $k$.
During UL WIT, device $k$ transmits its information signal to the AP with transmit power $p_k$. {In addition to the transmit power, each device also consumes a constant circuit power accounting for the power needed to operate its transmit filter, mixer, frequency synthesizers, etc., denoted by $p_{{\rm c},k}\geq 0$ \cite{zhang2016fundamental,qing15_wpcn_twc,qing16_wpcn_twc}.}
For the multiple access scheme in UL WIT, we consider two schemes, i.e., TDMA and NOMA. For T-WPCN, device $k$ exclusively accesses the spectrum for a duration of $\tau_k$, while for N-WPCN, all the devices access the spectrum simultaneously for a duration of $\bar{\tau}_{\rm 1}$. Then, the energy consumed by device $k$ during UL WIT for T-WPCN and N-WPCN can be expressed as $(p_{k}+p_{{\rm c},k})\tau_k$ and $(p_{k}+p_{{\rm c},k})\bar{\tau}_{\rm 1}$, respectively. Denote  $\gamma_k= \frac{g_k}{\sigma^2}$ as the normalized UL channel gain of device $k$, where $\sigma^2$ is the additive white Gaussian noise power at the AP.  For convenience,  we assume that the normalized UL channel power gains are sorted in ascending order, i.e., $0< \gamma_1\leq \gamma_2 \cdot\cdot\cdot \leq \gamma_K$.

\subsection{T-WPCN and Problem Formulation}
For T-WPCN, the achievable throughput of device $k$  in bits/Hz  can be expressed as
\begin{align}\label{eq6}
r_k=\tau_k \log_2\left(1+p_k\gamma_k\right).
\end{align}
 Then, the system throughput of T-WPCN is given by
\begin{align}
R_{\rm TDMA}=\sum_{k=1}^{K}r_k=\sum_{k=1}^{K} \tau_k  \log_2(1+{p_k\gamma_k}).
\end{align}
Accordingly, the SE maximization problem is formulated as
\begin{subequations} \label{probm10}
\begin{align}\label{eq10}
 \mathop {\mathrm{maximze} }\limits_{{\tau_{0},\{\tau_{k}\},\{p_{k}\}} }~~ &\sum_{k=1}^{K} \tau_k\log_2\left(1+p_k\gamma_k\right) \\
\mathrm{s.t.} ~~~~~~~&  \left({p_k}+p_{{\rm{c}},k}\right)\tau_k\leq \eta_kP_{\rm E}h_k\tau_0, ~ \forall\, k, \label{eq401} \\
& \tau_{0}+\sum_{k=1}^{K}\tau_k\leq T_{\mathop{\max}},  \label{eq402} \\
& \tau_{0}\geq0, ~  \tau_k\geq  0,  ~p_k\geq  0, ~\forall\, k.  \label{eq403}
\end{align}
\end{subequations}
In problem (\ref{probm10}), (\ref{eq401}) is the energy causality constraint which ensures that the energy consumed for WIT does not exceed the total energy harvested during WET. (\ref{eq402}) and (\ref{eq403}) are the total time constraint and the non-negativity constraints on the optimization variables, respectively.
\subsection{N-WPCN and Problem Formulation}
For N-WPCN, since all the $K$ devices share the same spectrum, SIC is employed at the AP to eliminate multiuser interference \cite{ding2014performance}.  Specifically, for detecting the message of the $k$-th device, the AP first decodes the message of the $i$-th device, $\forall\, i<k$, and then removes this message from the received signal, in the order of $i=1, 2,...,k-1$. The message of the $i$-th user, $\forall\, i>k$, is treated as noise. Hence, the achievable throughput of device $k$  in bits/Hz in N-WPCN  can be expressed as
\begin{align}\label{eq6}
r_k=\bar{\tau}_{\rm 1}\log_2\left(1+\frac{p_k\gamma_k}{\sum_{i=k+1}^{K} p_i\gamma_i+1}\right).
\end{align}
Then, the system throughput of T-WPCN is given by
\begin{align}\label{eq6}
R_{\rm NOMA}=\sum_{k=1}^{K}r_k=\bar{\tau}_{\rm 1} \log_2\left(1+\sum_{k=1}^{K}p_k\gamma_k\right).
\end{align}
Accordingly, the SE maximization problem is formulated as
 \begin{subequations} \label{probm20}
\begin{align}
 \mathop {\mathrm{maximize} }\limits_{{\tau_{0}, \bar{\tau}_{\rm 1},\{p_{k}\} } }~~ &\bar{\tau}_{\rm 1}\log_2\left(1+\sum_{k=1}^{K}p_k\gamma_k\right)\\
\mathrm{s.t.} ~~~~~~~&  \left({p_k}+p_{{\rm{c}},k}\right)\bar{\tau}_{\rm 1}\leq \eta_kP_{\rm E}h_k\tau_0, ~ \forall\, k,  \label{eq201}\\
&\tau_{0}+\bar{\tau}_{\rm 1}\leq T_{\mathop{\max}}, \label{eq202}\\
&\tau_{0}\geq0, ~ \bar{\tau}_{\rm 1}\geq  0, ~p_k\geq  0, ~\forall\, k.\label{eq203}
\end{align}
\end{subequations}
Similar to problem (\ref{probm10}), (\ref{eq201}), (\ref{eq202}), and (\ref{eq203}) represent the energy causality constraint, total time constraint, and non-negativity constraints, respectively.

\section{T-WPCN  or N-WPCN for IoT Networks?}
In this section, we first derive the optimal solutions to problems (\ref{probm10}) and (\ref{probm20}), respectively.  Then, we theoretically analyze and compare the system energy consumed and the SE achieved by both T-WPCN and N-WPCN.
\subsection{Optimal Solution for T-WPCN}
It can be shown that  each device will deplete all of its energy at the optimal solution, i.e., constraint (\ref{eq401}) holds with equality, since otherwise $p_k$ can be always increased to improve the objective value such that (\ref{eq401})  is active. Thus, problem (\ref{probm10}) is simplified to the following
  \begin{subequations}\label{probm30}
\begin{align}
 \mathop {\mathrm{maximize} }\limits_{{\tau_{0},\{\tau_{k}\} }}~~ &\sum_{k=1}^{K} \tau_k\log_2\left( 1-p_{{\rm{c}},k}\gamma_k+\frac{\eta_kP_{\rm E}h_k\gamma_k }{\tau_k}\tau_0\right)  \\
\mathrm{s.t.} ~~~
&\tau_{0}+\sum_{k=1}^{K}\tau_k\leq T_{\mathop{\max}}, \label{eq301} \\
& \tau_{0}\geq0, ~  \tau_k\geq  0, ~\forall\, k. \label{eq302}
\end{align}
  \end{subequations}
It is easy to verify that problem (\ref{probm30}) is a convex optimization problem and also satisfies the Slater's condition. Thus, the optimal solution can be obtained efficiently by applying the Lagrange dual method. To this end, we need the Lagrangian function of  problem (\ref{probm30}) which can be written as
\begin{align}
\mathcal{L}(\tau_0, \{\tau_k\}) = & \sum_{k=1}^{K} \tau_k\log_2\left(1-p_{{\rm{c}},k}\gamma_k+\frac{\eta_kP_{\rm E}h_k\gamma_k }{\tau_k} \tau_0\right) \nonumber\\
&+\lambda\left(T_{\max} -\tau_0-\sum_{k=1}^{K}\tau_k\right),
\end{align}
where $\lambda \geq 0$ is the Lagrange multiplier associated with (\ref{eq301}). (8c) is naturally satisfied since the PB is activated in the DL and each user is scheduled in the UL. Taking the partial derivative of $\mathcal{L}$ with respect to $\tau_0$ and $\tau_k$, respectively, yields

\begin{align}
\frac{\partial\mathcal{L}}{\partial \tau_0}&= \sum_{k=1}^{K} \frac{\eta_kP_{\rm E}h_k\gamma_k\log_2(e)}{ 1-p_{{\rm c},k}\gamma_k+x_k}-
\lambda,\\
\frac{\partial\mathcal{L}}{\partial \tau_k}&=\log_2\left(1-p_{{\rm c},k}\gamma_k +x_k\right)
                                                                                            -  \frac{x_k\log_2(e)}{1-p_{{\rm c},k}\gamma_k+x_k }-\lambda,
\end{align}
where $x_k = \frac{\eta_kP_{\rm E}h_k\gamma_k}{\tau_k}\tau_0, \forall\,k$.
Since $\tau_0>0$ and $\tau_k>0$, $\forall\,k$, always hold at the optimal solution, we have  $\frac{\partial\mathcal{L}}{\partial \tau_0}=0$ and $\frac{\partial\mathcal{L}}{\partial \tau_k}=0$, $\forall\, k$. As a result, the optimal values of $x_k$, $\forall\, k$, can be obtained by solving the following set of equations
 \begin{align}\label{eq_activate3011}
   \small
\mathcal{G}_k(x^*_k) \triangleq &\log_2(1-p_{{\rm{c}},k}\gamma_k +x^*_k)-\frac{x^*_k\log_2(e)}{1-p_{{\rm{c}},k}\gamma_k+x^*_k} \nonumber\\
&-\sum_{k=1}^{K}\frac{\eta_kP_{\rm{E}}h_k\gamma_k\log_2(e)}{1-p_{{\rm{c}},k}\gamma_k+x^*_k}=0, \forall\, k.
 \end{align}
{ Note that the first two terms of $\mathcal{G}_k(x^*_k)$ monotonically increase with $x^*_k$ while the last term is  the same for all users.  Thus,  $x^*_k$ can be efficiently obtained by the bisection method.}
 It can be shown that (\ref{eq301}) is active at the optimal solution, i.e., $\tau_0 +\sum_{k=1}^{K}\tau_k =  \tau_0 +\sum_{k=1}^{K} \frac{P_{\rm E}h_k\eta_k\gamma_k}{x^*_k}\tau_0=T_{\max}$.
With $x^*_k$, $\forall\, k$, from (\ref{eq_activate3011}), the optimal time allocation for T-WPCN is given by
 \begin{align}
   \small
 \tau^*_0& =\frac{T_{\rm max} }{1 +\sum_{k=1}^{K}\frac{\eta_kP_{\rm E}h_k\gamma_k }{x^*_k}},  \label{tau0}\\
  \tau^*_k& = \frac{\eta_kP_{\rm E}h_k\gamma_k}{x^*_k}\tau^*_0, \forall\, k. \label{tauk}
 \end{align}

\subsection{Optimal Solution for N-WPCN}
Similarly, problem (\ref{probm20}) can be simplified to the following problem:

  \begin{subequations}\label{probm40}
  \small
\begin{align}
 \mathop {\mathrm{maximize} }\limits_{{\tau_{0},\bar{\tau}_{\rm 1} } }~~ &\bar{\tau}_{\rm 1}\log_2\left(1-\sum_{k=1}^{K}p_{{\rm{c}},k}\gamma_k+\frac{\sum_{k=1}^{K}\eta_kP_{\rm E}h_k\gamma_k}{\bar{\tau}_{\rm 1}}\tau_0\right) \\
\mathrm{s.t.} ~~~&\tau_{0}+\bar{\tau}_{\rm 1}\leq T_{\mathop{\max}},\\
& \tau_{0}\geq0, ~  \bar{\tau}_{\rm 1}\geq  0.
\end{align}
  \end{subequations}

 It is interesting to observe that problem (\ref{probm40}) has the same structure as problem (\ref{probm30}) when $K=1$ with only minor changes in constant terms. As such, the proposed solution for T-WPCN can be immediately extended to N-WPCN.  Specifically, the optimal time allocation for N-WPCN is given by
  \begin{align}\label{tau00}
 \tau^{\star}_0 =\frac{T_{\rm max} }{1 +\frac{\sum_{k=1}^{K}\eta_kP_{\rm E}h_k\gamma_k }{x^{\star}}}, \,
 \bar{\tau}^{\star}_1 = \frac{\sum_{k=1}^{K}\eta_kP_{\rm E}h_k\gamma_k}{x^{\star}}\tau^{\star}_0,
 \end{align}
 where $x^{\star}$ is the unique root of
 {\small
 \begin{align}\label{eq_activate404}
\mathcal{G}(x^{\star}) \triangleq & \log_2\left(1-\sum_{k=1}^{K}p_{{\rm{c}},k}\gamma_k +x^{\star}\right)-\frac{x^{\star}\log_2(e)}{1-\sum_{k=1}^{K}p_{{\rm{c}},k}\gamma_k+x^{\star}} \nonumber\\
&-\frac{\sum_{k=1}^{K}\eta_kP_{\rm{E}}h_k\gamma_k\log_2(e)}{1-\sum_{k=1}^{K}p_{{\rm{c}},k}\gamma_k+x^{\star}}=0.
 \end{align}
}

The solutions  proposed in  Sections III-A and B serve as the theoretical foundation for the comparison between T-WPCN and N-WPCN.

\subsection{TDMA versus NOMA}
For notational simplicity, we first denote by $E^*_{\rm TDMA}$ and  $E^{\star}_{\rm NOMA}$  the total energy consumption of T-WPCN and N-WPCN at the optimal solutions to problems (\ref{probm30}) and (\ref{probm40}), respectively. The corresponding SEs are denoted by $R^*_{\rm TDMA}$ and  $R^{\star}_{\rm NOMA}$, respectively. 
 \begin{theorem}\label{WET_time}
 At the optimal solution, 1)  the DL WET time of N-WPCN in (\ref{tau00}) is greater than or equal to that of T-WPCN in (\ref{tau0}), i.e., $\tau^{\star}_0\geq  \tau^*_0$; 2) the energy consumption of N-WPCN is larger than or equal to that of T-WPCN, i.e.,
 \begin{align}
 E^{\star}_{\rm NOMA} \geq E^*_{\rm TDMA}, 
 \end{align}
  where ``='' holds when $p_{{\rm c},k}=0$, $\forall\, k$.
 \end{theorem}
 \begin{proof}
 Since $\sum_{k=1}^{K}p_{{\rm{c}},k}\gamma_k\geq p_{{\rm{c}},k}\gamma_k$, it is easy to show that $x^{\star}\geq x^*_k$, $\forall\, k$, from (\ref{eq_activate404}) and (\ref{eq_activate3011}), where ``='' holds when $p_{{\rm c},k}=0$, $\forall\, k$. Then, it follows from (\ref{tau00}) and (\ref{tau0}) that  $\tau^{\star}_0\geq  \tau^*_0$. Furthermore, since each device depletes all of its harvested energy, then the total energy consumption of N-WPCN and T-WPCN satisfies $E^{\star}_{\rm NOMA} = P_{\rm E}\tau^{\star}_0 \geq E^*_{\rm TDMA} = P_{\rm E}\tau^*_0$.
 \end{proof}

Theorem \ref{WET_time} implies that N-WPCN is more energy demanding than T-WPCN in terms of the total energy consumption. This is fundamentally due to simultaneous transmissions of multiple devices during UL WIT, which thereby leads to a higher circuit energy consumption. Furthermore,
 since $\tau^{\star}_0\geq  \tau^*_0$, more energy is also wasted during DL WET for N-WPCN than for T-WPCN. Next, we compare the SE of the two networks.

 \begin{theorem}\label{thm:rate_comparsion}
 The maximum SE of T-WPCN is greater than or equal to that of N-WPCN, i.e.,
\begin{align}\label{eq4.23}
R^*_{\rm TDMA}\geq R^{\star}_{\rm NOMA},
\end{align}
where ``='' holds when $p_{{\rm c},k}=0$, $\forall\, k$. 
 \end{theorem}
 \begin{proof}
Assume that $\{\tau^{\star}_{0},\bar{\tau}^{\star}_{1}\}$ achieves the maximum SE of problem (\ref{probm40}), $R^{\star}_{\rm NOMA}$. {Then, we can construct a new solution $\{\widetilde{\tau}_0,\{{\widetilde{\tau}}_{k}\}\}$ satisfying $\widetilde{\tau}_0 =\tau^{\star}_{0} $ and $\sum_{k=1}^{K} \widetilde{\tau}_k =\bar{\tau}^{\star}_1$ such that all devices achieve the same signal-to-noise ratio (SNR) in T-WPCN, i.e.,
\begin{align}\label{eq18}
\text{SNR}= &\frac{(\eta_kP_{\rm E}h_k\widetilde{\tau}_0-p_{{\rm{c}},k}\widetilde{\tau}_k)\gamma_k}{\widetilde{\tau}_k}= \frac{(\eta_mP_{\rm E}h_m\widetilde{\tau}_0-p_{{\rm{c}},m}\widetilde{\tau}_m)\gamma_m}{\widetilde{\tau}_m} \nonumber \\
=&\frac{ \sum_{k=1}^{K}(\eta_kP_{\rm E}h_k\widetilde{\tau}_0-p_{{\rm{c}},k}\widetilde{\tau}_k)\gamma_k}{ \sum_{k=1}^{K}\widetilde{\tau}_k}, \forall\, m\neq k.
\end{align}
It can be  verified that the constructed solution always exists and is also feasible for problem (\ref{probm30}). Denote the SEs achieved by the optimal solution $\{\tau^*_0,\{\tau^*_{k}\}\}$ and the constructed solution $\{\widetilde{\tau}_0,\{{\widetilde{\tau}}_{k}\}\}$ as $R^*_{\rm TDMA}$ and $\widetilde{R}_{\rm TDMA}$, respectively. Then, it follows that
{\small
\begin{align}
R^*_{\rm TDMA} & \geq \widetilde{R}_{\rm TDMA} \nonumber\\
& =  \sum_{k=1}^{K}\widetilde{\tau}_k\log_2\left( 1+\frac{(\eta_kP_{\rm E}h_k\widetilde{\tau}_0-p_{{\rm{c}},k}\widetilde{\tau}_k)\gamma_k}{\widetilde{\tau}_k}\right) \nonumber\\
& =  \sum_{k=1}^{K}\widetilde{\tau}_k\log_2\left( 1+\frac{ \sum_{m=1}^{K}(\eta_mP_{\rm E}h_m\widetilde{\tau}_0-p_{{\rm{c}},m}\widetilde{\tau}_m)\gamma_m }{ \sum_{m=1}^{K}\widetilde{\tau}_m}\right) \nonumber\\
&\overset{(a)}\geq \bar{\tau}^{\star}_1\log_2\left( 1+\frac{ \sum_{m=1}^{K}(\eta_mP_{\rm E}h_m{\tau}^{\star}_0-p_{{\rm{c}},m}\bar{\tau}^{\star}_1)\gamma_m }{ \bar{\tau}^{\star}_1} \right) \nonumber\\
&= R^{\star}_{\rm NOMA},
\end{align}}where  inequality ``$(a)$'' holds due to $\sum_{k=1}^{K} \widetilde{\tau}_k =\bar{\tau}^{\star}_1$ and $0<\widetilde{\tau}_k<\bar{\tau}^{\star}_1$, $\forall\, k$, and the equality holds when $p_{{\rm c},k}=0$, $\forall\, k$. } Thus, if $\exists\, k, p_{{\rm c},k}>0$, it follows that $R^*_{\rm TDMA}> R^{\star}_{\rm NOMA}$.
 Next, we prove that when  $p_{{\rm c},k}=0$, $\forall\, k$, the constructed solution is the optimal solution to problem (\ref{probm30}), i.e.,  $\tau^*_0=\widetilde{\tau}_0$ and  $\tau^*_k=\widetilde{\tau}_k$. 
The SE of T-WPCN is given by
\begin{align}
R_{\rm TDMA}& = \sum_{k=1}^{K} \tau_k\log_2\left( 1 +\frac{\eta_kP_{\rm E}h_k\gamma_k }{\tau_k}\tau_0\right) \nonumber \\
&\overset{(b)}\leq  \sum_{k=1}^{K} \tau_k\log_2\left( 1 +\frac{\sum_{m=1}^{K}\eta_mP_{\rm E}h_m\gamma_m }{\sum_{m=1}^{K}\tau_m}\tau_0\right)\nonumber \\
& = (1-\tau_0) \log_2\left( 1 +\frac{\sum_{m=1}^{K}\eta_mP_{\rm E}h_m\gamma_m }{1-\tau_0}\tau_0\right) \nonumber \\
&\overset{(c)} \leq  (1-\tau^{\star}_0) \log_2\left( 1 +\frac{\sum_{m=1}^{K}\eta_mP_{\rm E}h_m\gamma_m }{1-\tau^{\star}_0}\tau^{\star}_0\right)  \nonumber \\
&= R^{\star}_{\rm NOMA},
\end{align} where ``$(b)$'' holds due to the concavity of the logarithm function and ``='' holds when $\frac{\eta_kP_{\rm E}h_k\gamma_k }{\tau_k}\tau_0= \frac{\eta_mP_{\rm E}h_m\gamma_m }{\tau_m}\tau_0$, $\forall\, k$, which is exactly the same as (\ref{eq18}) for $p_{{\rm c},k}=0$, $\forall\, k$. Thus, we have $\tau^{*}_k =\widetilde{\tau}_k$.
 Equality  in ``$(c)$''  is due to the optimality of $\bar{\tau}^{\star}_0$ for N-WPCN. Thus, it follows that $\tau^*_0= \bar{\tau}^{\star}_0= \widetilde{\tau}_0$.
\end{proof}

Theorem \ref{thm:rate_comparsion} answers the question raised in the introduction regarding to the SE comparison of T-WPCN and N-WPCN. Specifically, TDMA in general achieves a higher SE than NOMA for wireless powered IoT devices.
This seems contradictory to the conclusions of previous works, e.g. \cite{ding2014performance}, which have shown that NOMA always outperforms OMA schemes such as TDMA. Such a conclusion, however, was based on the conventional transmit power limited scenario where more transmit power is always beneficial for improving the SE by leveraging SIC. To show this, suppose that the transmit power of device $k$ is limited by $p_k$ and the energy causality constraints in (\ref{probm10}) are removed.
By setting $\tau_0=0$ in (\ref{eq402}), we have
\begin{align}
R_{\rm TDMA}&=\sum_{k=1}^{K} \tau_k  \log_2(1+{p_k\gamma_k}) \nonumber\\
&\overset{(d)} \leq \sum_{k=1}^{K} \tau_k  \log_2\left(1+\sum_{m=1}^{K}{p_m\gamma_m}\right)\nonumber \\
&= T_{\max} \log_2\left(1+\sum_{k=1}^{K}{p_k\gamma_k}\right)= R_{\rm NOMA},
\end{align}
where strict inequality ``$(d)$'' holds if $p_k>0, \forall\,k$. Accordingly, $E_{\rm TDMA}=\sum_{k=1}^{K}\tau_kp_k \leq \sum_{k=1}^{K}T_{\max}p_k  = T_{\max}\sum_{k=1}^{K}p_k= E_{\rm NOMA}$.
{ This suggests that the potential SE gain achieved by  NOMA depends on the considered scenario.  When each user has a maximum transmit power limitation $p_k$, which we refer to as transmit power limited scenario, all users would transmit at $p_k$ for the entire duration $T_{\max}$.  The resulting SE gain of NOMA is at the expense of a higher energy consumption as shown above. On the other hand, if the total available energy of each device is constrained, which we refer to as energy limited scenario,
NOMA provides no SE gain over TDMA as shown in Theorem \ref{thm:rate_comparsion}, which is consistent with the observations in \cite{diamantoulakis2016wireless,chingoska2016resource}. More importantly, when the circuit power consumption is taken into account for practical IoT devices, NOMA achieves a strictly lower SE than TDMA. Recall that the key principle of NOMA for enhancing the SE is to allow devices to access the same spectrum simultaneously. This,  however, inevitably leads to a higher circuit energy consumption for NOMA because of the longer transmission time compared to TDMA, which is particularly detrimental to  IoT devices that are energy limited in general.}



\section{Numerical Results }
There are $10$ IoT devices randomly and uniformly distributed inside a disc with the PB in the center. {The carrier frequency is $750$ MHz and the bandwidth is $180$ kHz as in typical NB-IoT systems \cite{wang2017primer}. The reference distance is $1$ meter and the maximum service distance is $5$ meters \cite{bi2016wireless}. The AP is located $50$ meters away from the PB. Both the DL and UL channel power gains are modeled as $10^{-3}\rho^2d^{-\alpha}$  \cite{ju14_throughput}, where $\rho^2$ is an exponentially distributed random variable (i.e., Rayleigh fading is assumed) with unit mean and $d$ is the link distance.} The path loss exponent is set as $\alpha=2.2$.
 Without loss of generality, it is assumed that all IoT devices have identical parameters which are set as $\eta_k=0.9$ and $p_{{\rm c},k}=0.1$\,mW, $\forall\,k$ \cite{martins2017energy}. Other important parameters are set as $\sigma^2=-117$\,dBm, $P_{\rm E}=40$\,dBm, and $T_{\max}=0.1$\,s.




\subsection{SE versus PB Transmit Power}
Fig. \ref{pb}  shows the achievable throughput and energy consumption versus the PB transmit power, respectively. For comparison, two baseline schemes adopting TDMA and NOMA respectively are considered,  where $\tau_0=\frac{T_{\max}}{2}$ is set for both of them. This corresponds to the case that only $E^h_k=\frac{\eta_kP_{\rm E}h_kT_{\max}}{2}$ Joule of energy is available for device $k$, i.e., energy constrained IoT networks. Yet, the UL WIT is still optimized for maximizing the SE.
In Fig. \ref{pb} (a),  the throughputs of both T-WPCN and N-WPCN improve with $P_{\rm E}$. This is intuitive since with larger $P_{\rm E}$, the wireless powered IoT devices are able to harvest more energy during DL WET and hence achieve a higher throughput in UL WIT.  In addition, the baseline schemes suffer from a throughput loss for both TDMA and NOMA compared to the corresponding optimal scheme due to the fixed time allocation for DL WET, which implies that optimizing the DL WET duration is also important for maximizing the SE of wireless powered IoT networks. Furthermore, as suggested by Theorem 2, T-WPCN outperforms N-WPCN significantly and the performance gap between them becomes larger as $P_{\rm E}$ increase.
This is because larger $P_{\rm E}$ will reduce the DL WET time and thereby leave more time for UL WIT. Since all the devices in N-WPCN are scheduled simultaneously for UL WIT, the circuit energy consumption will be significantly increased compared to that of T-WPCN, which thus leads to a larger performance gap. Fig. \ref{pb} (b) shows that N-WPCN  is in general more energy demanding compared to T-WPCN for the optimal scheme, which verifies our theoretical finding in Theorem 1. Since $\tau_0=\frac{T_{\max}}{2}$ is set for both baseline schemes, they have the same total energy consumption. In addition, when $P_{\rm E}=28$ dBm, the energy consumption of optimal N-WPCN is close to that of optimal T-WPCN, which implies that each device $k$, $\forall\, k$, basically harvests a similar amount of energy in the DL of T-WPCN and N-WPCN. As such, the substantial SE loss in Fig. \ref{pb} (a) indicates that a significant portion of the harvested energy is consumed by the circuit rather than for signal transmission, due to the simultaneous transmission feature of NOMA.
\begin{figure}[!t]\vspace{-0.1cm}
\centering
\subfigure[Throughput comparison.]{\includegraphics[width=1.7in, height=1.4in]{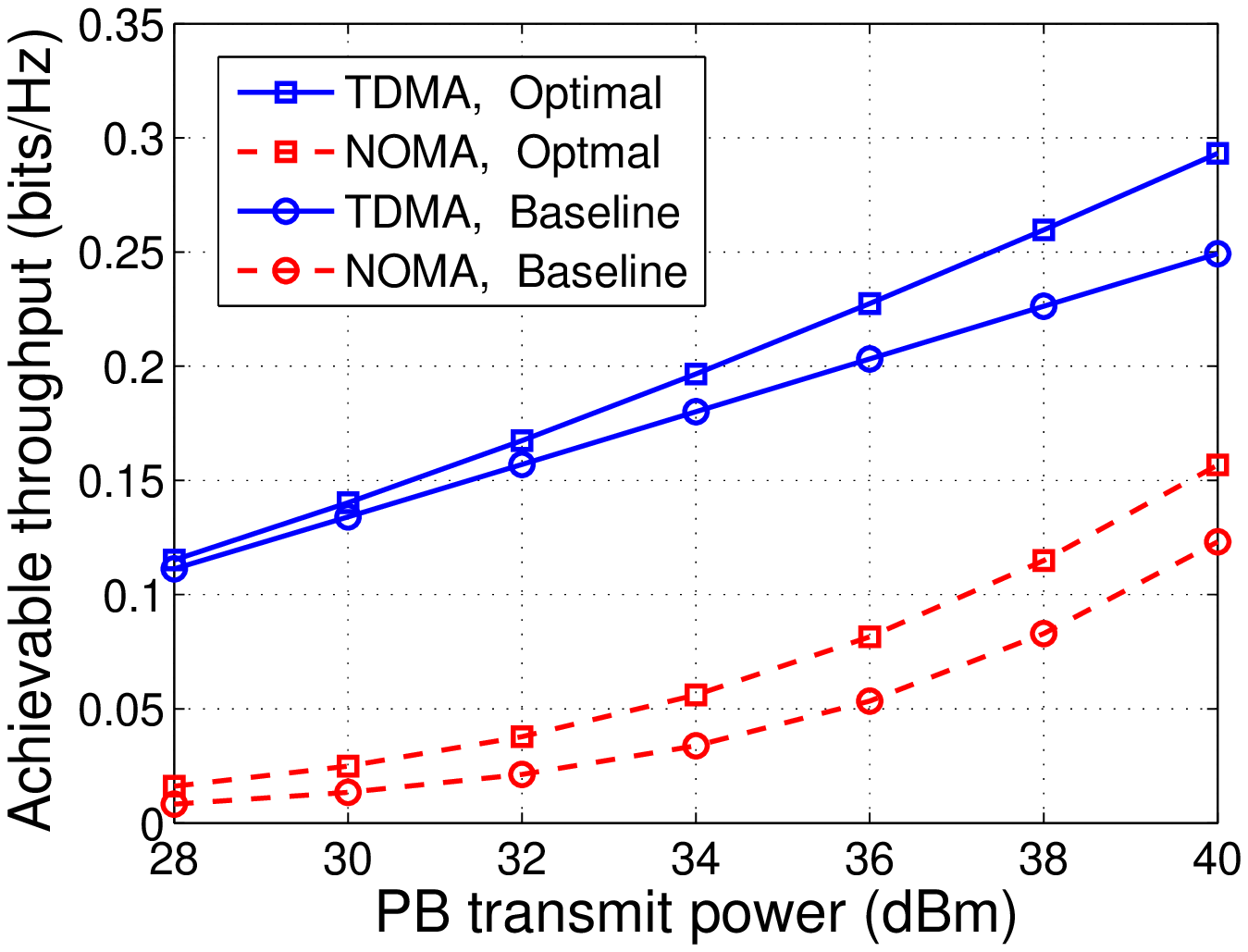}}
\subfigure[Energy consumption comparison.]{\includegraphics[width=1.7in, height=1.4in]{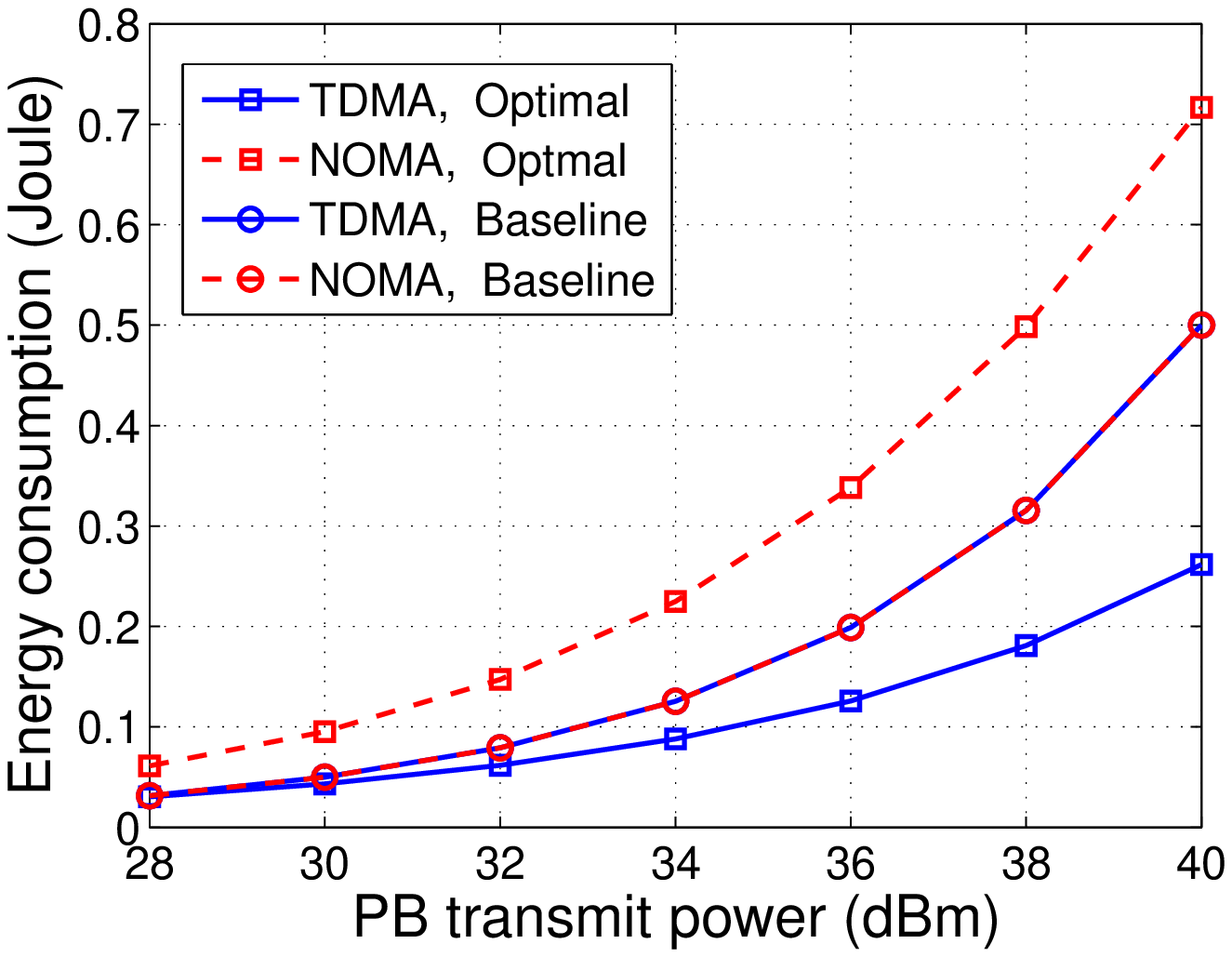}}
\caption{Throughput and energy consumption versus PB transmit power. } \label{pb}\vspace{-0.2cm}
\end{figure}
\subsection{SE versus Device Circuit Power}
Fig. \ref{pc} depicts the throughput and energy consumption versus the device circuit power consumption, respectively. Several observations are made as follows. First, for $p_{{\rm c},k}=0$ in Fig. \ref{pc} (a) and (b), T-WPCN and N-WPCN achieve the same throughput and energy consumption for $K=10$ and $K=50$, which coincides with our findings in Theorems 1 and 2. Second, for $K=10$ and $K=50$, the throughput and energy consumption for T-WPCN moderately decreases and increases with  $p_{{\rm c},k}$, respectively, while that for N-WPCN decreases and increases sharply with $p_{{\rm c},k}$, respectively. This suggests that the performance of N-WPCN is sensitive to $p_{{\rm c},k}$. In fact, for T-WPCN, when a device suffers from a worse DL channel condition, the corresponding harvested energy is also less. Then, the device will be allocated a short UL WIT duration such that the energy causality constraint is satisfied. However, for N-WPCN, since all devices transmit in the UL simultaneously, to meet the energy causality of all the devices, i.e., $\left({p_k}+p_{{\rm{c}},k}\right)\bar{\tau}_{\rm 1}\leq \eta_kP_{\rm E}h_k\tau_0= \eta_kP_{\rm E}h_k(1-\bar{\tau}_{\rm 1}), ~ \forall\, k$, it follows that $\bar{\tau}_{\rm 1}\leq \frac{ \eta_kP_{\rm E}h_k}{{p_k}+p_{{\rm{c}},k}+\eta_kP_{\rm E}h_k}\leq \frac{ \eta_kP_{\rm E}h_k}{p_{{\rm{c}},k}+\eta_kP_{\rm E}h_k} $, $ \forall\, k$. As can be seen, the UL WIT duration $\bar{\tau}_{\rm 1}$ is always limited by the worst DL channel gain of all devices for $p_{{\rm{c}},k}>0$, a phenomenon which we refer to as ``worst user bottleneck problem". In addition, concurrent transmissions also lead to higher circuit energy consumption. As a result, the throughput and energy consumption of N-WPCN are significantly reduced and increased, respectively, as $p_{{\rm c},k}$ increases. Third, given the ``worst user bottleneck problem", it is expected that when $K$ increases from $10$ to $50$, the performance of N-WPCN decreases in both Fig. \ref{pc} (a) and (b). In contrast, for T-WPCN, since the UL WIT duration of each user can be individually allocated based on the DL and UL channel gains of each device,  multiuser diversity can be exploited to improve the performances as  $K$ increases from $10$ to $50$.

\begin{figure}[!t]\vspace{-0.1cm}
\centering
\subfigure[Throughput comparison.]{\includegraphics[width=1.7in, height=1.4in]{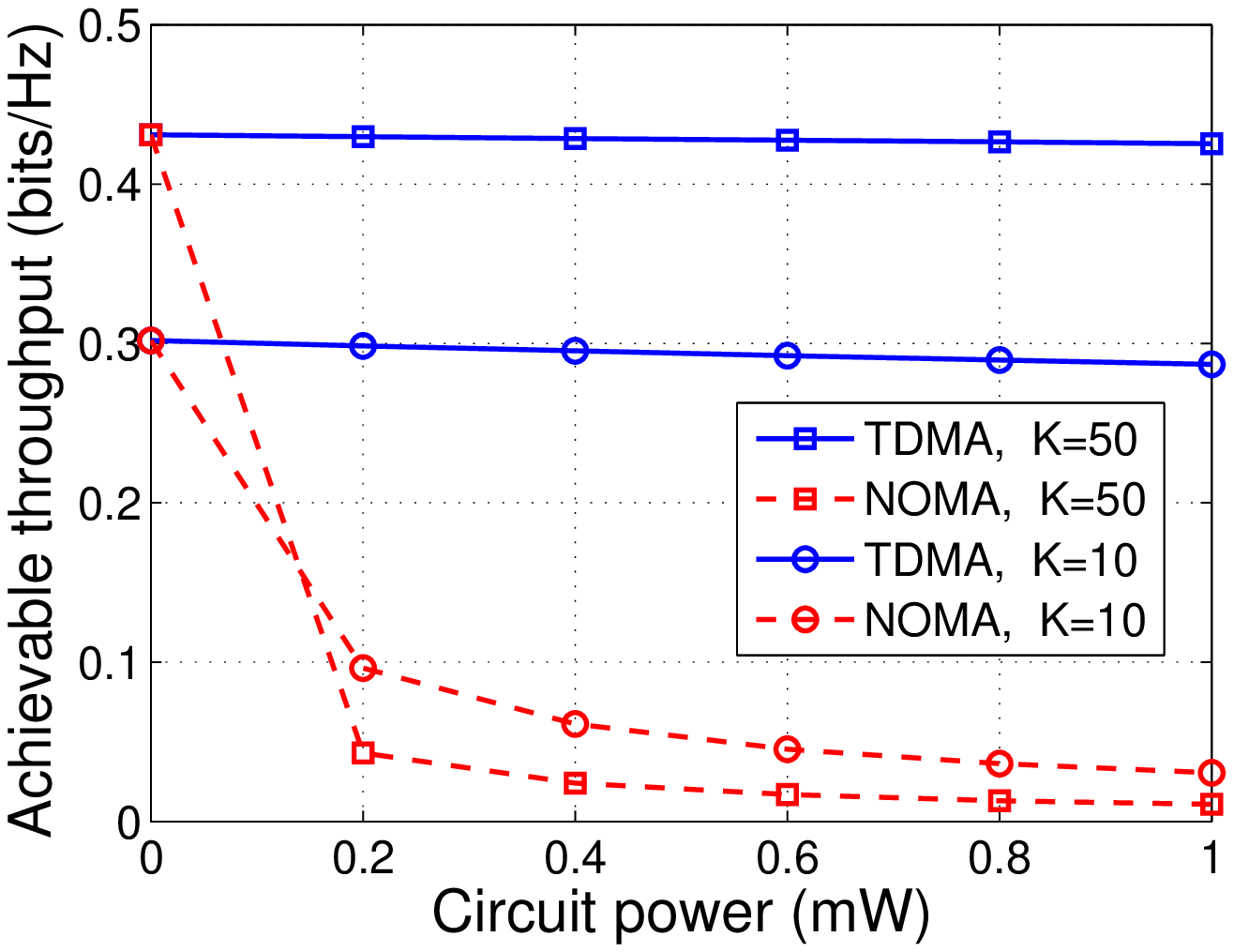}}
\subfigure[Energy consumption comparison.]{\includegraphics[width=1.7in, height=1.4in]{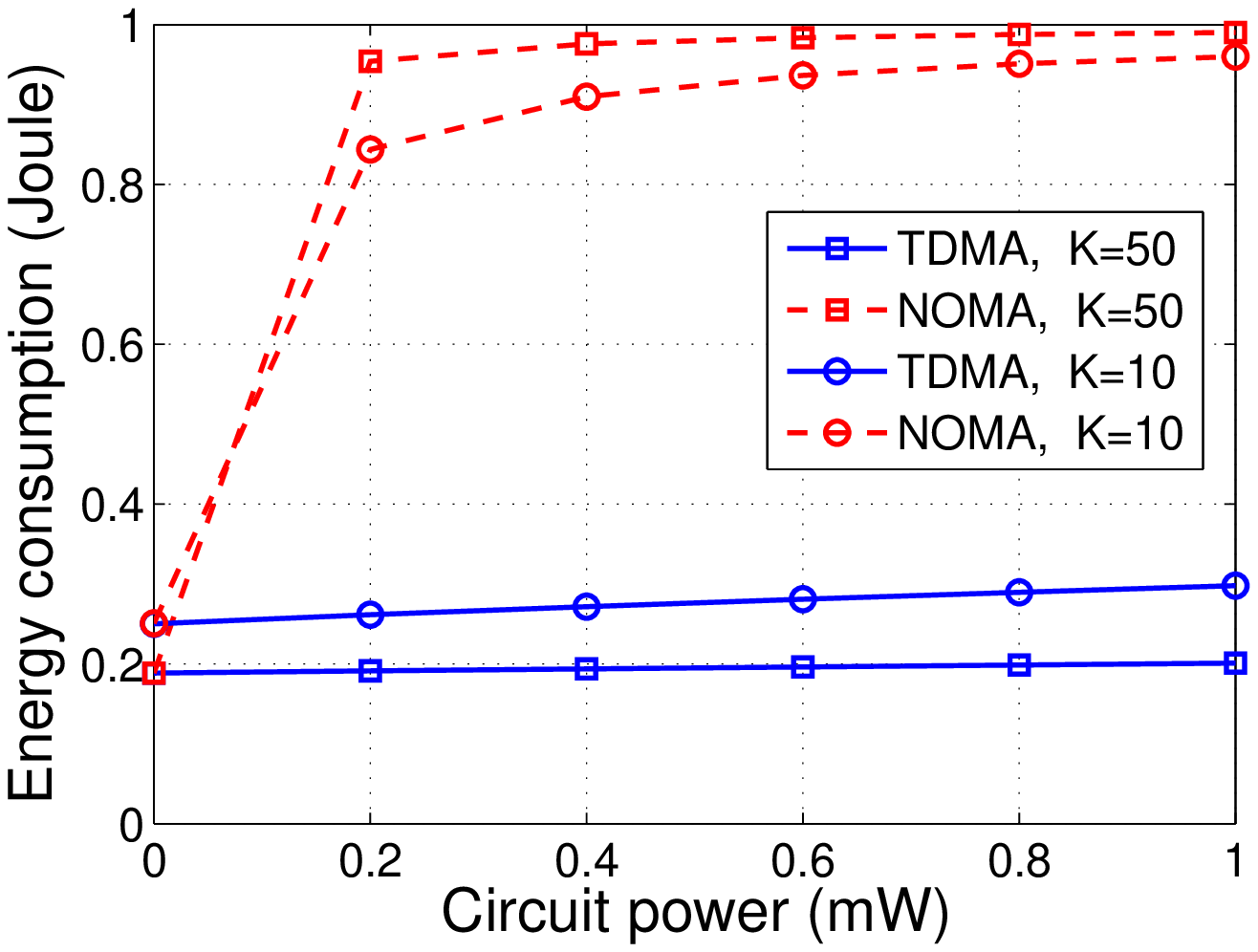}}
\caption{Throughput and energy consumption versus device circuit power. } \label{pc}\vspace{-0.2cm}
\end{figure}

%

\section{Conclusions}
In this paper, we have answered a fundamental question: Does NOMA improve SE and/or  reduce the total energy consumption of the wireless powered IoT networks? By taking  into account the circuit energy consumption of the IoT devices, we have found that N-WPCN is neither spectral efficient nor energy efficient, compared to T-WPCN. 
 This suggests that NOMA may  not be a practical solution for spectral and energy efficient wireless IoT networks with energy constrained devices.  The case with user fairness consideration is an interesting topic for future work.
 In addition, the results in the paper also suggest that energy-efficient NOMA transmission is a promising research direction that is worth studying.



\bibliographystyle{IEEEtran}
\bibliography{IEEEabrv,mybib}
\end{document}